\documentclass{llncs}
\pdfoutput=1
\usepackage[utf8]{inputenc}
\usepackage[T1]{fontenc}
\usepackage{lmodern}
\usepackage[activate={true,nocompatibility},final,tracking=true,kerning=true,spacing=true]{microtype}
\usepackage{amsmath, amssymb}
\usepackage{mathtools}
\usepackage{booktabs}
\usepackage{xspace}
\usepackage[inline]{enumitem}
\usepackage{wrapfig}
\usepackage{tikz}
\usepackage[hidelinks]{hyperref}

\usetikzlibrary{arrows, shapes}

\setlist*[enumerate]{label={(\alph*)}, ref={\alph*)}}

\newcommand{\cspace}{\mathcal{C}}

\newcommand{\hypspace}{\mathcal{H}}

\newcommand{\sspace}{\mathcal{S}}
\newcommand{\sord}{\sqsubseteq_\mathrm{s}}
\newcommand{\sjoin}{\sqcup}
\newcommand{\bigsjoin}{\bigsqcup}
\newcommand{\sbot}{\bot_\mathrm{s}}
\newcommand{\concrete}{\gamma}
\newcommand{\consistent}{\kappa}
\newcommand{\learner}{\lambda}
\newcommand{\ordinals}{\mathbb{O}}

\newcommand{\teacher}{\tau}
\newcommand{\target}{\mathcal{T}}
\newcommand{\seq}[1]{\langle #1 \rangle}
\newcommand{\complexord}{\preceq}
\newcommand{\wqo}{\preceq}
\newcommand{\W}{\mathcal{W}}
\newcommand{\A}{\mathcal{A}}

\newcommand{\nat}{\mathbb{N}}

\newcommand{\exampleend}{{\tikz \draw (0, 0) -| +(1ex, 1ex);}}

\title{Abstract Learning Frameworks for Synthesis}
\author{Christof Löding\inst{1} \and P. Madhusudan\inst{2} \and Daniel Neider\inst{2}$^{,}$\inst{3}}
\institute{RWTH Aachen
\and University of Illinois, Urbana-Champaign
\and University of California, Los Angeles}

\begin{document}

\maketitle

\begin{abstract}
We develop abstract learning frameworks  for synthesis that embody the principles
of the CEGIS (counterexample-guided inductive synthesis) algorithms in current literature.
Our framework is based on iterative learning
from a hypothesis space that captures synthesized objects, using counterexamples from an
abstract sample space, and a concept space
that abstractly defines the semantics of synthesis. We show that a variety of
synthesis algorithms in current literature can be embedded in this general framework.
We also exhibit three general recipes for convergent synthesis: the first two recipes based on
finite spaces and Occam learners generalize all techniques of convergence used in existing engines,
while the third, involving well-founded quasi-orderings, is  new,
and we instantiate it to concrete synthesis problems. 
\end{abstract}

\section{Introduction}

The field of synthesis, which includes several forms of synthesis including
synthesizing controllers~\cite{PnueliR89}, program expressions~\cite{sketch}, program repairs~\cite{DBLP:conf/cav/KneussKK15}, program translations~\cite{DBLP:journals/debu/CheungMSAM14,DBLP:conf/oopsla/KaraivanovRV14},
loop invariants~\cite{ICEML,DBLP:conf/cav/0001LMN14}, and even entire programs~\cite{MannaW80,DBLP:conf/aaip/Kitzelmann09}, has become a fundamental 
and vibrant subfield in programming languages. While classical studies of synthesis
have focused on synthesizing entire programs or controllers from specifications~\cite{MannaW80,PnueliR89},
there is a surge of tractable methods that have emerged in recent years in 
synthesizing small program expressions. These expressions often are complex but small, and
are applicable in niche domains such as program sketching~\cite{sketch} (finding program expressions that
complete code), synthesizing Excel programs for string transformations~\cite{DBLP:conf/popl/Gulwani11}, synthesizing superoptimized code~\cite{DBLP:conf/asplos/Schkufza0A13}, deobfuscating code~\cite{DBLP:conf/icse/JhaGST10},  
synthesizing invariants to help in verification~\cite{DBLP:conf/cav/0001LMN14,ICEML}, etc.

One prominent technique that has emerged in recent years for expression synthesis is
based on \emph{inductively learning expressions from samples}. Assume the synthesis problem is 
to synthesize an expression $e$ that satisfies some specification $\psi(e)$. 
The crux of this approach is to \emph{ignore} the precise specification $\psi$, 
and instead synthesize an expression based on certain \emph{facets}
of the specification. 
These incomplete facets of the specification are often 
much simpler in structure and in logical complexity compared to the specification,
and hence synthesizing an expression satisfying the constraints the facets impose is 
more tractable. The learning-based approach to synthesis hence happens in rounds--- 
in each round, the learner synthesizes an expression that satisfies the current facets,
and a \emph{verification oracle} checks whether the expression satisfies the actual specification $\psi$,
and if not, finds a new facet of the specification witnessing this. The learner then continues
to synthesize by adding this new facet to its collection. 

This counter-example 
guided inductive synthesis (CEGIS) approach~\cite{armando_thesis} to synthesis in current literature philosophically 
advocates precisely this kind of inductive synthesis.
The CEGIS approach has emerged as a powerful technique in several domains
of both program synthesis as well as program verification ranging from synthesizing program invariants
for verification~\cite{DBLP:conf/cav/0001LMN14,ICEML} to specification mining~\cite{AlurCMN05}, program expressions that complete sketches~\cite{sketch}, 
superoptimization~\cite{DBLP:conf/asplos/Schkufza0A13}, control~\cite{DBLP:conf/hybrid/JinDDS13}, string transformers for spreadsheets~\cite{DBLP:conf/popl/Gulwani11}, protocols~\cite{DBLP:conf/pldi/UdupaRDMMA13}, 
etc.

The goal of this paper is to develop a \emph{theory of iterative learning-based synthesis} through a formalism
we call \emph{abstract learning frameworks for synthesis}. The framework we develop aims to be general and abstract,
encompassing several known CEGIS frameworks as well as several other synthesis algorithms not generally
viewed as CEGIS. The goal of this line of work is to build a framework, with accompanying concepts, 
definitions, and vocabulary that can be used to understand and combine learning-based synthesis across different
domains.

An abstract learning framework (ALF) (see Figure~\ref{fig:alf} on Page~\pageref{fig:alf}) consists of three spaces: $\hypspace$, $\sspace$, and
$\cspace$. The \emph{(semantic) concept space} $\cspace$ gives semantic descriptions of the concepts that we wish to synthesize, 
the \emph{hypotheses space} $\hypspace$ comprises restricted (typically syntactically restricted) forms of the concepts
to synthesize, and the sample space $\sspace$ consists of samples (modeling facets of the specification) from which the learner synthesizes
hypotheses.  The spaces $\hypspace$ and $\sspace$ are related by a variety of functions that give semantics to samples and semantics to 
hypotheses using the space $\cspace$. The conditions imposed on these relations capture the learning problem
precisely, and their abstract formulation facilitates modeling a variety of synthesis frameworks in the literature.

The target for synthesis is specified as a \emph{set} of semantic concepts.
This is an important digression from classical learning frameworks, where often one can  assume that there is a
\emph{particular} target concept that the learner is trying to learn. Note that in synthesis problems, we must
\emph{implement the teacher as well}, and hence the modeling of the target space is important.
In synthesis problems, the teacher does not have a single target in mind nor does she know explicitly the target set 
(if she knew, there would be no reason to synthesize!). 
Rather, she knows the \emph{properties} that capture the set of target concepts. 
For instance, in invariant synthesis, the teacher
knows the properties of a set being an invariant for a loop, and this defines implicitly a \emph{set} of invariants
as target.
The teacher needs to examine a hypothesis and
check whether it satisfies the properties defining the target set. Consequently, we can view the teacher as a
\emph{verification oracle} that checks whether a hypothesis belongs to the implicitly defined target set. 


We exhibit a variety of existing synthesis frameworks that can be naturally seen as instantiations of 
our abstract-learning framework, where the formulation shows the diversity in the instantiations of 
the spaces. These include (a) a variety of CEGIS-based synthesis techniques for synthesizing program expressions in
sketches (completing program sketches~\cite{sketch}, synthesizing loop-free programs~\cite{DBLP:conf/pldi/GulwaniJTV11}, mining specifications~\cite{DBLP:conf/hybrid/JinDDS13}, synthesizing synchronization code for concurrent programs~\cite{DBLP:conf/cav/CernyCHRRST15}, 
etc.),  (b) synthesis from input-output examples such as Flashfill~\cite{DBLP:conf/popl/Gulwani11},
(c) the CEGIS framework applied to the concrete problem of solving synthesis problems expressed in the SMT-based SyGuS format~\cite{DBLP:conf/fmcad/AlurBJMRSSSTU13,DBLP:series/natosec/AlurBDF0JKMMRSSSSTU15}, 
and three synthesis engines that use learning to synthesize solutions, 
(d) invariant synthesis frameworks, including Houdini~\cite{houdini} and the more recent ICE-learning
model for synthesizing loop invariants~\cite{DBLP:conf/cav/0001LMN14}, spanning a variety of domains from arithmetic~\cite{DBLP:conf/cav/0001LMN14,ICEML} to quantified invariants over data structures~\cite{CAVQDA},
and (e) synthesizing fixed-points and abstract transformers in abstract interpretation
settings~\cite{thakur}.

Formalizing of synthesis algorithms as ALFs
can help highlight the nuances of different learning-based synthesis algorithms, even for the \emph{same} problem.
One example comprises two inductive learning approaches for synthesizing program invariants--- one 
based on the ICE learning model~\cite{DBLP:conf/cav/0001LMN14},
and the second which is any synthesis engine for logically specified synthesis problems in the SyGuS format,
which can express invariant synthesis.  Though both can be seen as CEGIS-based synthesis algorithms, the sample space
for them are very different, and hence the synthesis algorithms are also different--- the significant performance
differences between SyGuS-based solvers and ICE-based solvers(the latter performing better) in the recent SyGuS competition 
(invariant-synthesis track) 
suggest that this choice may be crucial~\cite{sygus15competition}.
Another example are two classes of CEGIS-based solvers for synthesizing linear integer arithmetic functions against
SyGuS specifications--- one based on a sample space that involves purely inputs to the function being synthesized~\cite{DBLP:conf/cav/Saha0M15,alchemist2},
while the other is the more standard CEGIS algorithm based on valuations of quantified variables.

We believe that just describing an approach as a  learning-based synthesis algorithm or a CEGIS algorithm does not convey
the nuances of the approach--- it is important to precisely spell out the sample space and the semantics of this space with
respect to the space of hypotheses being learned. The ALF framework gives the vocabulary in phrasing these nuances, allowing
us to compare and contrast different approaches. 



\noindent \textbf{Convergence:~~}
The second main contribution of this paper is to study \emph{convergence} issues in the general abstract learning-based framework
for synthesis. We first show that under the reasonable assumptions that the learner is consistent (always proposes a hypothesis consistent with
the samples it has received) and the teacher is honest (gives a sample that distinguishes the current hypothesis from the target set
without ruling out any of the target concepts), the iterative learning will always converge \emph{in the limit} 
(though, not necessarily in finite time, of course). This theorem vouches for the correctness of our abstract
formalism in capturing abstract learning, and utilizes all the properties that define ALFs.

We then turn to studying strategies for convergence in finite time.
We propose three general techniques for ensuring successful termination for the learner.
First, when the hypothesis space is bounded, it is easy to show that any consistent learner 
(paired with an honest teacher) will converge in finite time. Several examples of these exist in learning--- 
learning conjunctions as in the Houdini algorithm~\cite{houdini}, etc., learning Boolean functions (like decision-tree learning with purely Boolean predicates as
attributes) or functions over bit-vector domains (Sketch~\cite{sketch} and the SyGuS solvers that work on bit-vectors), 
and learning invariants using
specialized forms of a finite class of automata that capture list/array invariants~\cite{CAVQDA}.

The second recipe is a formulation of the Occam's razor principle that uses parsimony/simplicity as the learning bias~\cite{sep_simplicity}.
The idea of using Occam's principle in learning is prevalent~(see Chapter 2 of~\cite{KearnsV94} and \cite{mitchell}) though its universal appeal in generalizing concepts is debatable~\cite{DBLP:journals/datamine/Domingos99}.
We show, however, that learning using Occam's principle helps in convergence.
A learner is said to be an Occam learner if there is a \emph{complexity ordering}, which needs to be a total quasi order where the set of 
elements below any element is finite, such that the learner always learns \emph{a smallest} concept according to this order that
is consistent with the sample. We can then show that any Occam learner will converge to some target concept, if one exists, in finite time.
This result generalizes many convergent learning mechanisms that we know of in the literature (for example, the convergent ICE-learning
algorithms for synthesizing invariants using constraint solvers~\cite{DBLP:conf/cav/0001LMN14}, and the enumerative solvers in almost every domain
of synthesis~\cite{DBLP:conf/icalp/Kuncak14,DBLP:conf/cav/KneussKK15,DBLP:conf/pldi/OseraZ15,DBLP:conf/pldi/UdupaRDMMA13}, including for SyGuS~\cite{DBLP:conf/fmcad/AlurBJMRSSSTU13,DBLP:series/natosec/AlurBDF0JKMMRSSSSTU15}, that enumerate by dovetailing through expressions).

The first two recipes for finite convergence cover all the methods we know in the literature for convergent learning-based synthesis, 
to the best of our knowledge.
The third recipe for finite convergence is a more complex one based on well-founded quasi orderings. 
This recipe is involved and calls for using clever initial queries that force the teacher to divulge information that then
makes the learning space tractable.
We do not know of any existing synthesis learning frameworks that use this natural recipe, but
propose two new convergent learning algorithms following this recipe, one for intervals, and the other for conjunctive linear inequality
constraints over a set of numerical attributes over integers. 



\section{Abstract Learning Frameworks for Synthesis} \label{sec:alf}

In this section we introduce our abstract learning framework for
synthesis. Figure~\ref{fig:alf} gives an overview of the components
and their relations that are introduced in the following (ignore the
target $\target$, $\gamma^{-1}(\target)$, and the maps $\tau$ and $\lambda$ for now).
We explain these components in more detail after the formal definition.

\begin{wrapfigure}[11]{R}{6.75cm}
	\centering
	\vskip -2\baselineskip
	\begin{tikzpicture}[font=\scriptsize, auto, >=stealth]
		\node[draw, thick, fill=black!7, shape=circle, minimum size=1.5cm] (sample) at (0, 0) {};
		\node[anchor=south, text width=2.4cm, align=center] at (sample.north) {sample space $(\sspace, \sord, \sjoin, \sbot)$};

		\node[draw, thick, fill=black!7, shape=circle, minimum size=1.5cm] (hypothesis) at (4, 0) {};
		\node[anchor=south, text width=2.5cm, align=center] at (hypothesis.north) {hypothesis space $\mathcal H$};
		\node[draw, fill=black!50, text=white, rounded corners=1pt, align=center] at (hypothesis) {$\concrete^{-1}(\mathcal T)$};

		\node[draw, thick, fill=black!7, shape=circle, minimum size=1.5cm] (concept) at (2, -2) {};
		\node[anchor=west] at (concept.east) {concept space $\mathcal C$};
		\node[draw, shape=ellipse, minimum width=1cm, fill=black!50, text=white] at (concept) {$\mathcal T$};

		\draw[->] (sample) edge[bend right=15] node[text width=1.25cm, swap, align=center] {learner $\lambda \colon \sspace \to \mathcal H$} (hypothesis);
		\draw[->] (hypothesis) edge[bend right=15] node[text width=1.25cm, swap, align=center] {teacher $\tau \colon \mathcal H \to \sspace$} (sample);
		\draw[->] (hypothesis) edge node[text width=1.75cm] {concretization $\concrete \colon \mathcal H \to \mathcal C$} (concept);
		\draw[->] (sample) edge node[text width=1.75cm, swap]
                {consistency $\consistent \colon \sspace \to 2^\mathcal C$} (concept);

	\end{tikzpicture}
	
	\vskip -.75\baselineskip
	\caption{Components of an ALF} \label{fig:alf}
\end{wrapfigure}
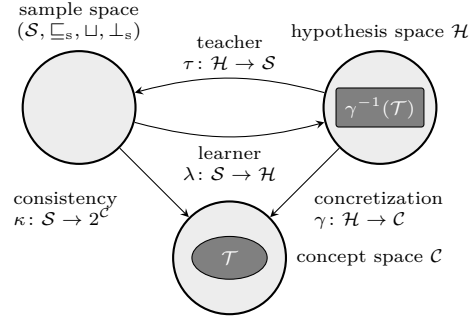



\begin{definition}[Abstract Learning Frameworks] \label{def:ALF}
An \emph{abstract learning framework for synthesis} (ALF, for short), 
is a tuple $\A = (\cspace, \hypspace, (\sspace,\sord,\sjoin,\sbot) , \concrete, \consistent)$, with
\begin{itemize}
  \item A class $\cspace$, called the concept space,
  \item A class $\hypspace$, called the hypothesis space, 
  \item A class $\sspace$, called the sample space, with a join semi-lattice $(\sspace,\sord,\sjoin,\sbot)$
              defined over it,
  \item A concretization function $\concrete: \hypspace \rightarrow \cspace$, and
  \item A consistency function: $\consistent: \sspace \rightarrow
  2^\cspace$ satisfying $\consistent(\sbot) = \cspace$ and $\consistent(S_1 \sjoin S_2) = \consistent(S_1) \cap \consistent(S_2)$ for all $S_1, S_2 \in \sspace$.  
  If the second condition is relaxed to $\consistent(S_1 \sjoin S_2) \subseteq \consistent(S_1) \cap \consistent(S_2)$, we speak of a \emph{general ALF}.

\end{itemize}

We say an ALF has a complete sample space if the sample space
$(\sspace,\sord,\sjoin,\sbot)$
is a complete join semi-lattice (i.e., if the join is defined for
arbitrary subsets of $\sspace$).
In this case, the consistency relation has to satisfy
$\consistent(\bigsjoin(\sspace')) = \bigcap_{S \in \sspace'}
\consistent(S)$ for each $\sspace' \subseteq \sspace$ (and $\consistent(\bigsjoin(\sspace')) \subseteq \bigcap_{S \in \sspace'}
\consistent(S)$ for general ALFs).
\end{definition}

As in computational learning theory, as presented
in e.g.,\ in \cite{Angluin92} or \cite{KearnsV94}, we consider a
\emph{concept space} $\cspace$, which contains the objects that we are
interested in. For example, in an invariant synthesis setting in verification, an element $C
\in \cspace$ would be a \emph{set} of program configurations. In the synthesis
setting, $\cspace$ could contain the objects we would like to synthesize, 
such as all functions from $\mathbb{Z}^n$ to $\mathbb{Z}$.

The \emph{hypothesis space} $\hypspace$ contains the objects that the
learner produces. These are representations of (some) elements from
the concept space. 
For example, if $\cspace$ consists of 
all functions from $\mathbb{Z}^n$ to $\mathbb{Z}$, then $\hypspace$ could 
consist the set of all functions expressible in linear arithmetic.

The relation between hypotheses and concepts is given by a
concretization function $\concrete: \hypspace \rightarrow \cspace$ that
maps hypotheses to concepts (their semantics). 

In classical computational learning theory for classification~\cite{KearnsV94,mitchell}, one often considers samples
consisting of positive and negative examples. 
If learning is used to infer a target concept that is not
uniquely defined but rather should satisfy certain properties,  then
samples consisting of positive and negative examples are sometimes not
sufficient. As we will show later, samples can be quite complex (see Section~\ref{sec:examples} for such examples,
including implication counterexamples and grounded formulas). 

We work with a \emph{sample space},
which is a bounded join-semilattice $(\sspace,\sord,\sjoin,\sbot)$
(i.e., $\sord$ is a partial order over $\sspace$ with $\sbot$ as the least element,
and $\sjoin$ is the binary least upper-bound operator on $\sspace$ with respect to this ordering).
An element $S \in \sspace$, when given by the teacher, intuitively, gives some information about a target specification. 
The join is used by the learner to combine the samples returned as feedback by the teacher 
during iterative learning. The least element $\sbot$ corresponds to the empty sample.
We encourage the reader to think of the join as the union of samples.

The \emph{consistency relation} $\consistent$ captures the semantics
of samples with respect to the concept space by assigning to each
sample $S$ the set $\consistent(S)$ of concepts that are consistent
with the sample. 
The first condition on $\kappa$
says that all concepts are consistent with the empty sample $\sbot$.
The second condition says that the set of samples consistent with the
join of two samples is precisely the set of concepts that is
consistent with both the samples.  Intuitively, this means that
joining samples does not introduce new inconsistencies, and existing
inconsistencies transfer to bigger samples.
The condition that
$\consistent(S_1 \sjoin S_2) \subseteq \consistent(S_1) \cap
\consistent(S_2)$ is natural, as it says that if a concept is
consistent with the join of two samples, then the concept must be
consistent with both of them individually. The condition that
$\consistent(S_1 \sjoin S_2) \supseteq \consistent(S_1) \cap
\consistent(S_2)$ is debatable; it claims that samples when taken
together cannot eliminate a concept that they couldn't eliminate
individually.  We therefore mention the notion of \emph{general ALF}
in Definition~\ref{def:ALF}. However, we have not found any natural
example that requires such a generalization, and therefore prefer to
work with ALFs instead of general ALFs in the rest of the paper. In
Definition~\ref{def:teacher}, we comment on what needs to be adapted to
make the results of the paper go through for general ALFs.

The following simple but useful observation on monotonicity of the
consistency relation (which easily follows from the property of the
consistency relation, for ALFs and general ALFs) is used in some
of the proofs.

\begin{remark}\label{rem:consistency-monotonic}
If $S_1 \sord S_2$, then $\consistent(S_2) \subseteq
\consistent(S_1)$.
\end{remark}

Some other auxiliary definitions we will need: 
We define $\consistent_\hypspace(S) := \{H \in \hypspace \mid \concrete(H) \in \consistent(S)\}$ to be the set of hypotheses that are
consistent with $S$.  For a sample $S \in \sspace$ we say
that \emph{$S$ is realizable} if there exists a hypothesis that is
consistent with $S$ (i.e., $\consistent_\hypspace(S) \not= \emptyset$).

\paragraph{ALF Instances and Learners}
An instance of a learning task for an ALF is given by a specification that defines 
target concepts. The goal is to infer a hypothesis whose semantics is
such a target concept. In classical computational learning theory,
this target is a unique concept. In applications for synthesis,
however, there can be many possible target concepts, for example, all
inductive invariants of a program loop.

Formally, a \emph{target specification} is just a set $\target
\subseteq \cspace$ of concepts. An ALF instance combines an ALF
and a target specification:

\begin{definition}[ALF Instance]
An ALF instance is a pair $(A,\target)$ where 
$\A = (\cspace, \hypspace, (\sspace,\sord,\sjoin,\sbot) , \concrete, \consistent)$ 
is an ALF and $\target \subseteq \cspace$ is a target specification.
\end{definition}
\medskip

The goal of learning-based synthesis is for the learner to synthesize
\emph{some} element $H \in \hypspace$ such that $\concrete(H) \in \target$.
Furthermore, the role of the teacher is to instruct the learner giving
reasons why the hypothesis produced by the learner in the current round
does not belong to the target set.

There is a subtle point here worth emphasizing. In synthesis frameworks,
the teacher does not explicitly know the target space $\target$.
Rather she knows a definition of the target space, and she can examine
a hypothesis $H$ and check whether it satisfies the properties required of
the target set. For instance, when synthesizing an invariant for a
program, the teacher knows the properties of the invariant
(inductiveness, etc.)  and gives counterexample samples based on
failed properties.

We say that the target specification is \emph{realizable by a
  hypothesis}, or simply realizable, if there is some $H
\in \hypspace$ with $\concrete(H) \in \target$.  For a hypothesis $H
\in \hypspace$, we often write  $H \in \target$ instead of
$\concrete(H) \in \target$.

As in classical computational learning theory, we define a
\emph{learner} (see Figure~\ref{fig:alf}) to be a function that maps samples to hypotheses,
and a consistent learner to be a learner that only proposes consistent hypotheses
for samples.

\begin{definition}
A learner for an ALF $\A = (\cspace, \hypspace, (\sspace,\sord,\sjoin,\sbot) , \concrete, \consistent)$
is a map $\learner : \sspace \rightarrow \hypspace$ that assigns a hypothesis to every sample.
A \emph{consistent learner} is a learner $\learner$ with $\concrete(\learner(S)) \in \consistent(S)$ for all realizable samples
$S \in \sspace$.
\end{definition}

%
\paragraph{Iterative learning.}
In the iterative learning setting, the learner produces a hypothesis
starting from some initial sample (e.g., $\sbot$). For each hypothesis
provided by the learner that does not satisfy the target specification,
a teacher (see Figure~\ref{fig:alf}) provides feedback by returning a sample witnessing that the
hypothesis does not satisfy the target specification. 



\begin{definition} \label{def:teacher}
Let $(A, \target)$ be an ALF instance with $\A = (\cspace, \hypspace, (\sspace,\sord,\sjoin,\sbot) , \concrete, \consistent)$,
and $\target \subseteq \cspace$. A \emph{teacher} for this ALF
instance is a function $\teacher: \hypspace \rightarrow \sspace$
that satisfies the following two properties:
\begin{description}
\item[i) Progress:] $\teacher(H) = \sbot$ for each target element $H \in \target$,
  and $\concrete(H) \notin \consistent(\teacher(H))$ for all $H \notin
  \target$, and 
\smallskip
\item[ii) Honesty:] $\target \subseteq \consistent(\teacher(H))$ for each $H
  \in \hypspace$.\kern-.06em\footnote{For general ALFs one has to require that the least upper bound of all samples returned by the teacher is consistent with all targets (and for non-complete sample lattices the least upper bound of all possible finite sets of samples returned by the teacher).}
\end{description}
\end{definition}

Firstly, progress says that if the hypothesis is in the target set, then the teacher
must return the ``empty'' sample $\sbot$, signaling that the learner has learned a target; otherwise, the teacher must return a sample that rules out the current hypothesis. This ensures that a consistent learner
can never propose the same hypothesis again, and hence makes progress.
Secondly, honesty demands that the sample returned by the teacher is 
consistent with \emph{all} target concepts. This ensures that the teacher does not eliminate any
element of the target set arbitrarily. 

When the learner and teacher interact iteratively, the learner produces a sequence of hypotheses,
where in each round it proposes a hypothesis $\learner(S)$ for the current sample $S
\in \sspace$, and then adds the feedback $\teacher(\learner(S))$ of
the teacher to obtain the new sample. 

\begin{definition}
Let $(A, \target)$ be an ALF instance with $\A = (\cspace, \hypspace,
(\sspace,\sord,\sjoin,\sbot) , \concrete, \consistent)$, and
$\target \subseteq \cspace$.  Let
$\learner: \sspace \rightarrow \hypspace$ be a learner, and let
$\teacher: \hypspace \rightarrow \sspace$ be a teacher.  The combined
behavior of the learner $\learner$ and teacher $\teacher$ is the
function $f_{\teacher, \learner} : \sspace \rightarrow \sspace$, where
$f_{\teacher, \learner}(S) := S \sjoin \teacher(\learner(S))$.

The sequence of hypotheses generated by the learner $\learner$ and teacher $\teacher$
is the transfinite sequence
 $\seq{S_{\teacher, \learner}^\alpha \mid \alpha \in
  \ordinals}$, where $\ordinals$ denotes the class of all ordinals, obtained by
iterative application of $f_{\teacher, \learner}$:
\begin{itemize}
\item $S_{\teacher,\learner}^0 := \sbot$;
\item $S_{\teacher,\learner}^{\alpha+1} :=
  f_{\teacher,\learner}(S_{\teacher,\learner}^\alpha)$ for successor
  ordinals; and
\item $S_{\teacher,\learner}^\alpha := \bigsjoin_{\beta < \alpha}
  S_{\teacher,\learner}^\beta$ for limit ordinals.
\end{itemize}

If the sample lattice is not complete, the above definition is
restricted to the first two items and yields a sequence indexed by
natural numbers.
\end{definition}


The following lemma states that the teacher's properties of progress
and honesty transfer to the iterative setting for consistent learners
if the target specification is realizable. 

\begin{lemma} \label{lem:progress-honesty}
Let $\target$ be realizable, $\learner$ be a consistent learner, and $\teacher$ be a teacher.
If $\sspace$ is a complete sample lattice, then
\begin{enumerate}
\item the learner makes progress: for all $\alpha \in \ordinals$,
  either $\consistent(S_{\teacher,\learner}^\alpha) \supsetneq
  \consistent(S_{\teacher,\learner}^{\alpha+1})$ and
  $\learner(S_{\teacher,\learner}^\alpha) \notin
  \consistent(S_{\teacher,\learner}^{\alpha+1})$, or
  $\learner(S_{\teacher,\learner}^\alpha) \in \target$, and
\item the sample sequence is consistent with the target specification:
  $\target \subseteq \consistent(S_{\teacher,\learner}^\alpha)$ for
  all $\alpha \in \ordinals$.
\end{enumerate}
If $\sspace$ is a non-complete sample lattice, then (a) and (b) hold
for all $\alpha \in \nat$.
\end{lemma}
\begin{proof}
The proof is a straight-forward transfinite induction, using the
properties of the teacher and the consistency relation. For the case
of non-complete sample lattice, ignore the limit step in the proof
below.

For part (b), the induction base is given by $\consistent(\sbot) =
\cspace$. The induction step for limit ordinals directly follows from
the property of the consistency relation: all previous samples are
consistent with the target specification, so their join is, too.\footnote{For general
  ALFs one uses the property that the least upper bound  of all
  samples returned by the teacher is consistent with all targets.} 
 For
a successor ordinal $\alpha+1$ it follows from the fact that
$S_{\teacher,\learner}^{\alpha+1}$ is a join of two samples that are
both consistent with the target specification.

For part (a), let $\alpha \in \ordinals$ such that $H :=
\learner(S_{\teacher,\learner}^\alpha) \notin \target$. Then
$S_{\teacher,\learner}^{\alpha+1} = S_{\teacher,\learner}^\alpha \cup
S$ with $S = \teacher(H)$. In particular,
$S_{\teacher,\learner}^\alpha \sord S_{\teacher,\learner}^{\alpha+1}$
and thus $\consistent(S_{\teacher,\learner}^\alpha) \supseteq
\consistent(S_{\teacher,\learner}^{\alpha+1})$ by
Remark~\ref{rem:consistency-monotonic}. For the strictness of the
inclusion, note that $\concrete(H) \in
\consistent(S_{\teacher,\learner}^\alpha)$ because $\learner$ is a
consistent learner (and $S_{\teacher,\learner}^\alpha$ is realizable
because it is consistent with the realizable target specification by
(b)). Furthermore, $\concrete(H) \notin \consistent(S)$ by the
progress property of the teacher, and hence $\concrete(H) \notin
\consistent(S_{\teacher,\learner}^{\alpha+1})$.  \qed
\end{proof}

We end with an example of an ALF. Consider the problem of synthesizing
guarded affine functions that capture how a piece of code $P$ behaves,
as in program deobfuscation.\label{ex:deobfuscation} Then the concept
class could be all functions from $\mathbb{Z}^n$ to $\mathbb{Z}$, the
hypothesis space would be the set of all expressions describing a
guarded affine function (in some fixed syntax). The target set (as a subset
of ${\cal C}$) would
consist of a \emph{single} function $\{f_t\}$, where $f_t$ is the
function computed by the program $P$. For any hypothesis function $h$,
let us assume we can build a teacher who can compare $h$ and $P$ for
equivalence, and, if they differ, return a counterexample of the form
$(\vec{i},o)$, which is a concrete input $\vec{i}$ on which $h$ differs
from $P$, and $o$ is the output of $P$ on $\vec{i}$.  Then the sample
space would consist of sets of such pairs (with union for join and
empty set for $\sbot$). 
The set of functions consistent with a set of samples would be the
those that map the inputs mentioned in the samples to their appropriate
outputs.
The iterative learning will then model the
process of synthesis, using learning, a guarded affine function that
is equivalent to $P$.

\section{Convergence of iterative learning} \label{sec:convergence}

In this section, we study convergence of the iterative learning
process.  We start with a general theorem on transfinite
convergence (convergence in the limit) for complete sample lattices. 
We then turn to convergence in finite time and exhibit three recipes
that guarantee convergence.

From Lemma~\ref{lem:progress-honesty} one can conclude that the
transfinite sequence of hypotheses constructed by the learner
converges to a target set. 

\begin{theorem} \label{the:transfinite-convergence}
Let $\sspace$ be a complete sample lattice, $\target$ be realizable,
$\learner$ be a consistent learner, and $\teacher$ be a teacher. Then
there exists an ordinal $\alpha$ such that
$\learner(S_{\teacher,\learner}^\alpha) \in \target$.
\end{theorem}
\begin{proof}
Let $\alpha$ be an ordinal with cardinality bigger than $|\hypspace|$
(bigger than $|\sspace|$ also works). If
$\learner(S_{\teacher,\learner}^\beta) \not\in \target$ for all $\beta
< \alpha$, then Lemma~\ref{lem:progress-honesty}~(a) implies that all
$\learner(S_{\teacher,\learner}^\beta)$ for $\beta < \alpha$ are
pairwise different, which contradicts the cardinality assumption. \qed
\end{proof}

The above theorem ratifies the
choice of our definitions, and the proof (relying on Lemma~\ref{lem:progress-honesty}) crucially uses all aspects of our 
definitions (the honesty and progress properties of the teacher, the condition imposed on $\kappa$ in an ALF, the notion of consistent
learners, etc.).


Convergence in finite time is clearly the more desirable notion, and we
propose tactics for designing learners that converge in
finite time.
For an ALF instance $(\mathcal{A},\target)$, we say that a learner
$\learner$ \emph{converges for a teacher $\teacher$} if there is an $n
\in \nat$ such that $\learner(S_{\teacher,\learner}^n) \in \target$,
which means that $\learner$ produces a target hypothesis after $n$
steps. We say that $\learner$ converges if it converges for every 
teacher. We say that $\learner$ converges from a sample $S$ in case
the learning process starts from a sample $S \not= \sbot$ (i.e., if
$S_{\learner,\teacher}^0 = S$).



\subsubsection{Finite hypothesis spaces}\label{sec:finite_hypothesis_space}
We first note that if the hypothesis space (or the concept
space) is finite, then any consistent learner converges: by Lemma~\ref{lem:progress-honesty}, the learner always
makes progress, and hence never proposes two hypotheses that
correspond to the same concept. Consequently, 
the learner only produces a finite number of hypotheses before finding one
that is in the target (or declare that no such hypothesis exists).

There are several synthesis engines using learning that use finite hypothesis spaces. For example, Houdini~\cite{houdini} is a learner of \emph{conjunctions} over a fixed finite set of predicates and, hence, has a finite hypothesis space. Learning decision trees over purely Boolean attributes (not numerical)~\cite{DBLP:books/mk/Quinlan93} is also convergent because of finite hypothesis spaces, and this extends to the ICE learning model as well~\cite{ICEML}. Invariant generation for arrays and lists using \emph{elastic QDAs}~\cite{CAVQDA} also uses a convergence argument that relies on a finite hypothesis space.

\subsubsection{Occam Learners}\label{sec:occam_learner}
We now discuss the most robust strategy we know for convergence, based on the Occam's razor principle. Occam's razor advocates parsimony or simplicity~\cite{sep_simplicity},
that the simplest concept/theory that explains a set of observations is better, as a virtue in itself.
There are several learning algorithms that use parsimony as a learning bias in machine learning
(e.g., \emph{pruning} in decision-tree learning~\cite{mitchell}), though the general applicability of Occam's razor
in machine learning as a sound means to generalize is debatable~\cite{DBLP:journals/datamine/Domingos99}.
We now show that in \emph{iterative} learning, following Occam's principle leads
to convergence in finite time. However, the role of \emph{simplicity} itself is not the technical 
reason for convergence, but that there is \emph{some} ordering of concepts that biases the learning.


Enumerative learners are a good example of this. In enumerative learning, the
learner enumerates hypotheses in some order, and always conjectures the first
consistent hypothesis. In an iterative learning-based synthesis
setting, such a learner always converges on some target concept, if one exists, in finite time.

Requiring a total order of the hypotheses is in some situations too
strict. If, for example, the hypothesis space consists of
deterministic finite automata (DFAs), we could build a learner that
always produces a DFA with the smallest possible number of states that
is consistent with the given sample. However, the relation
$\complexord$ that compares DFAs w.r.t.\ their number of states is not
an ordering because there are different DFAs with the same number of
states. 

In order to capture such situations, we work with a \emph{total
  quasi-order} $\complexord$ on $\hypspace$ instead of a total
order. A quasi-order (also called preorder) is a transitive and
reflexive relation. The relation being total means that $H \complexord
H'$ or $H' \complexord H$ for all $H,H' \in \hypspace$. The difference
to an order relation is that $H \complexord H'$ and $H' \complexord H$
can hold in a quasi-order, even if $H \not= H'$.

In analogy to enumerations, we require that each hypothesis has only
finitely many hypotheses ``before'' it w.r.t.\ $\complexord$, as
expressed in the following definition.

\begin{definition}
A \emph{complexity ordering} is a total quasi-order $\complexord$ such
that for each $x \in \hypspace$ the set $\{y \in \hypspace \mid y \complexord
x\}$ is finite.
\end{definition}

The example of comparing DFAs with respect to their number of states is such a
complexity ordering.

\begin{definition}
A consistent learner that always constructs a smallest hypothesis with
respect to a complexity ordering $\preceq$ on $\hypspace$ is called an
\emph{$\complexord$-Occam learner}.
\end{definition}

\begin{example} \label{ex:complexity-ordering}
Consider $\hypspace = \cspace$ to be the interval domain over the integers
consisting of all intervals of the form $[l,r]$, where $l,r \in
\mathbb{Z} \cup \{-\infty,\infty\}$ and $l \le r$.  We define $[l,r]
\complexord [l',r']$ if either $[l,r] = [-\infty,\infty]$ or
$\max\{|x| \mid x \in \{l,r\} \cap \mathbb{Z}\} \le \max\{|x| \mid
x \in \{l',r'\} \cap \mathbb{Z}\}$.  For example, $[-4,\infty]
\complexord [1,7]$ because $4 \le 7$. 
This ordering $\complexord$ satisfies the property that for each
interval $[l,r]$ the set $\{[l',r'] \mid [l',r'] \complexord [l,r]\}$
is finite (because there are only finitely many intervals using
integer constants with a bounded absolute value). A standard
positive/negative sample $S= (P,N)$ with $P,N
\subseteq \nat$ is consistent with all intervals that
contain the elements from $P$ and do not contain an element from $N$.
A learner that maps $S$ to an interval that uses integers with the
smallest possible absolute value (while being consistent with $S$) is
an $\complexord$-Occam learner. For example, such a learner would map
the sample $(P=\{-2,5\},N=\{-8\})$ to the interval
$[-2,\infty]$. \hfill\exampleend
\end{example}

The next theorem shows that $\complexord$-Occam learners ensure
convergence in finite time. 

\begin{theorem} \label{the:occam-learner}
If $\target$ is realizable and $\learner$ is a $\complexord$-Occam
learner, then $\learner$ converges.  Furthermore, the
learner converges to a $\complexord$-minimal target element.
\end{theorem}
\begin{proof}
Pick any target element $T \in \hypspace$, which exists because
$\target$ is realizable. Since $\teacher$ is honest,
$T \in \consistent(S_{\teacher,\learner}^n)$ for all $n$ by
Lemma~\ref{lem:progress-honesty}(b). Thus, on the iterated sample
sequence, a $\complexord$-Occam learner never constructs an element
which is strictly above $T$ w.r.t.\ $\preceq$. Since there are only
finitely many hypothesis that are not strictly above $T$, and since
the learner always makes progress according to
Lemma~\ref{lem:progress-honesty}, it converges to a target element in
finitely many steps, which itself does not have any other target
elements below, and thus is $\complexord$-minimal. \qed
\end{proof}

There are several existing algorithms in the literature that use such
orderings to ensure convergence. Several enumeration-based solvers are
convergent because of the ordering of enumeration (e.g., the
generic enumerative solver for SyGuS
problems~\cite{DBLP:conf/fmcad/AlurBJMRSSSTU13,DBLP:series/natosec/AlurBDF0JKMMRSSSSTU15}).
The invariant-generation ranging over conditional linear arithmetic
expressions described in~\cite{DBLP:conf/cav/0001LMN14} ensures
convergence using a total quasi-order based on the number of
conditionals and the values of the coefficients.  The learner uses
templates to restrict the number of conditionals and a
constraint-solver to find small coefficients for linear constraints.




\subsubsection{Convergence using Tractable Well Founded Quasi-Orders}
The third strategy for convergence in finite time that we propose is based on well
founded quasi-orders, or simply well-quasi-orders. 
Interestingly, we know of no existing learning algorithms in the literature
that uses this recipe for convergence (a technique of similar flavor is used in \cite{Blum92}). We exhibit in this section a learning
algorithm for intervals and for conjunctions of inequalities of numerical attributes based on this recipe. 
A salient feature of this recipe is that the convergence actually uses the
samples returned by the teacher in order to converge (the first two
recipes articulated above, on the other hand, would even guarantee convergence if the teacher
just replies yes/no when asked whether the hypothesis is in the target set).

A binary relation $\wqo$ over some set $X$ is a well-quasi-order if it is transitive and
reflexive, and for each infinite sequence $x_0,x_1,x_2, \ldots$ there
are indices $i < j$ such that $x_i \wqo x_j$. In other words, there
are no infinite descending chains and no infinite anti-chains for
$\wqo$.

\begin{definition}
Let $(\A, \target)$ be an ALF instance with
$\A = (\cspace, \hypspace, (\sspace,\sord,\sjoin,\sbot) , \concrete, \consistent)$. A subset of hypotheses $\W \subseteq \hypspace$ is called
\emph{wqo-tractable} if 
\begin{enumerate}[nosep]
\item there is a  well-quasi-order $\wqo_{\W}$ on
$\W$, and
\item for each realizable sample $S \in \sspace$ with
  $\consistent_\hypspace(S) \subseteq \W$, there is some
  $\wqo_\W$-maximal hypothesis in $\W$ that is consistent with $S$.
\end{enumerate}
\end{definition}

\begin{example} \label{ex:wqo-tractable}
Consider again the example of intervals over $\mathbb{Z} \cup
\{-\infty, \infty\}$ with samples of the form $S = (P,N)$ (see
Example~\ref{ex:complexity-ordering}). Let $p \in \mathbb{Z}$ be a point
and let $\mathcal{I}_p$ be the set of all intervals that contain the
point $p$. Then, $\mathcal{I}_p$ is wqo-tractable with the standard
inclusion relation for intervals, defined by $[\ell,r] \subseteq [\ell',r']$
iff $\ell \ge \ell'$ and $r \le r'$. Restricted to intervals that
contain the point $p$, this is the product of two well-founded orders
on the sets $\{x \in \mathbb{Z} \mid x \le p\}$ and $\{x \in \mathbb{Z} \mid x \ge
p\}$, and as such is itself well-founded \cite[Theorem~2.3]{Higman52}.
Furthermore, for each realizable sample $(P,N)$, there is a unique
maximal interval over $\mathbb{Z} \cup \{-\infty,\infty\}$ that
contains $P$ and excludes $N$. Hence, the two conditions of
wqo-tractability are satisfied.
(Note that this ordering on the set of \emph{all} intervals is not a well-quasi-order;
 the sequence $[-\infty, 0], [-\infty, -1], [-\infty, -2], \ldots$ witnesses this.) \hfill\exampleend
\end{example}

On a wqo-tractable $\W \subseteq \hypspace$ a learner can ensure
convergence by always proposing a maximal consistent hypothesis, as
stated in the following lemma.
\begin{lemma} \label{lem:wqo-tractable}
Let $\target$ be realizable, $\W \subseteq \hypspace$ be wqo-tractable
with well-quasi-order $\wqo_\W$, and $S$ be a sample such that
$\consistent_\hypspace(S) \subseteq \W$. Then, there exists a learner that
converges from the sample $S$.
\end{lemma}
\begin{proof}
For any sample $S'$ with $S \sord S'$, the set
$\consistent_\hypspace(S')$ of hypotheses consistent with $S'$ is a
subset of $\W$. Therefore, there is some $\wqo_\W$-maximal element in
$\W$ that is consistent with $S'$. The strategy of the learner is to
return such a maximal hypothesis. Assume, for the sake of
contradiction, that such a learner does not converge from $S$ for some
teacher $\teacher$.  Let $H_0,H_1, \ldots$ be the infinite sequence of
hypothesis produced by $\learner$ and $\teacher$ starting from $S$,
and let $S_0,S_1,S_2, \ldots$ be the corresponding sequence of samples
(with $S_0 = S$).  The well-foundedness of $\wqo_\W$ implies that
there are $i < j$ with $H_i \wqo_\W H_j$. However, $S_i \sord S_j$
because $S_j$ is obtained from $S_i$ by joining answers of the
teacher. Therefore, $H_j$ is also consistent with $S_i$
(Remark~\ref{rem:consistency-monotonic}). This contradicts the
choice of $H_i$ as a maximal hypothesis that is consistent with $S_i$. \qed
\end{proof}
As shown in Example~\ref{ex:wqo-tractable}, for each $p \in
\mathbb{Z}$, the set $\mathcal{I}_p$ of intervals containing $p$ is
wqo-tractable. Using this, we can build a convergent learner starting
from the empty sample $\sbot$. First, the learner starts by proposing the empty interval,
the teacher must either confirm that this is a target or return a positive example, that is, a point $p$
that is contained in every target interval. Hence, the set of hypotheses
consistent with this sample is wqo-tractable and the learner can
converge from here on as stated in Lemma~\ref{lem:wqo-tractable}. 
%
In general, the strategy for the learner is to force in one step a sample $S$ such that the set $\consistent_\hypspace(S) =
\mathcal{I}_p$ is wqo-tractable. This is generalized in the following
definition. 

\begin{definition}
We say that an \emph{ALF is wqo-tractable} if there is a finite set
$\{H_1, \ldots, H_n\}$ of hypotheses such that $\kappa_\hypspace(S)$
is wqo-tractable for all samples $S$ that are inconsistent with all
$H_i$, that is, $\consistent_\hypspace(S) \cap
\{H_1, \ldots, H_n\} = \emptyset$. 
\end{definition}

As explained above, the interval
ALF is wqo-tractable with the set $\{H_1, \ldots, H_n\}$ consisting
only of the empty interval.

Combining all the previous observations, we obtain convergence for
wqo-tractable ALFs.

\begin{theorem} \label{the:wqo-learner}
For every ALF instance $(\A,\target)$ such that $\A$ is wqo-tractable
and  $\target$ is realizable, there is a convergent learner.
\end{theorem}
\begin{proof}
A convergent learner can be built as follows. Let $\{H_1, \ldots,
H_n\}$ be the finite set of hypotheses from the definition of
wqo-tractability of an ALF. 
\begin{itemize}
\item As long as the current sample is consistent with some $H_i$,
  propose such an $H_i$.

\item Otherwise, the current sample $S$ is such that
  $\consistent_\hypspace(S)$ is wqo-tractable, and thus the learner
  can apply the strategy from Lemma~\ref{lem:wqo-tractable}. \qed
\end{itemize}
\end{proof}


\paragraph{A convergent learner for  conjunctive linear inequality
constraints.}  We have illustrated wqo-tractability for intervals in
Example~\ref{ex:wqo-tractable}. We finish this section by showing that
this generalizes to higher dimensions, that is, to the domain of
$n$-dimensional hyperrectangles in $(\mathbb{Z} \cup
\{-\infty,\infty\})^n$, which form the hypothesis space in this
example. Each such hyperrectangle is a product of intervals over
$(\mathbb{Z} \cup \{-\infty,\infty\})^n$.
Note that hyperrectangles can, e.g., be used to model
conjunctive linear inequality constraints over a set $f_1, \ldots, f_n:
\mathbb{Z}^d \rightarrow \mathbb{Z}$ of numerical
attributes.

The sample space depends on the type of target specification that we
are interested in. We consider here the typical sample space of
positive and negative samples (however, the reasoning below also works
for other sample spaces, e.g., ICE sample spaces that additionally
include implications). So, samples are of the form $S = (P,N)$,
where $P,N$ are sets of points in $\mathbb{Z}^n$ interpreted as
positive and negative examples (as for intervals, see
Example~\ref{ex:complexity-ordering}).

The following lemma provides the ingredients for building a convergent
learner based on wqo-tractability.
\begin{lemma}\label{lem:wqo-hyperrectangle}
\begin{enumerate}
\item For each realizable sample $S = (P,N)$, there are maximal
hyperrectangles that are consistent with $S$ (possibly more than one).
\item For each $p \in \mathbb{Z}^n$, the set $\mathcal{R}_p$ of
hyperrectangles containing $p$ is well-quasi-ordered by inclusion.
\end{enumerate}
\end{lemma}
\begin{proof}
For the first claim, note that for each increasing chain $R_0
\subseteq R_1 \subseteq \cdots$ of hyperrectangles that are all
consistent with $S$, the union $R := \bigcup_{i \ge 0} R_i$ is also a
hyperrectangle that is consistent with $S$. More precisely, if $R_i =
[l_1^i,r_1^i] \times \cdots \times [l_n^i,r_n^i]$, then $R = [l_1,r_1]
\times \cdots \times [l_n,r_n]$ with $l_j = \inf\{l_j^i \mid i \ge 0\}$
and $r_j = \sup\{r_j^i \mid i \ge 0\}$.

Furthermore, if the chain is strictly increasing, then there exists
$j$ such that $l_j = -\infty$ and all $l_j^i \not= -\infty$, or $r_j =
\infty$ and all $r_j^i \not= \infty$. Hence, if $R$ itself can be
extended again into an infinite strictly increasing chain, then the
union of this chain will contain an additional $\infty$ of
$-\infty$. This can happen at most $2n$ times before reaching the
hyperrectangle containing all points, which is certainly
maximal. Thus, there has to be a maximal hyperrectangle consistent
with $S$. 


We now prove the second claim. For a point $p = (p_1, \ldots, p_n)$,
the set $\mathcal{R}_p$ is the product $\mathcal{R}_p =
\mathcal{I}_{p_1} \times \cdots \times \mathcal{I}_{p_n}$ of the sets
of intervals containing the points $p_i$.  Furthermore, the inclusion
order for hyperrectangles is the $n$-fold product of the inclusion
order for intervals. Thus, the inclusion order on $\mathcal{R}_p$ is a
well-quasi-order because it is a product of well-quasi-orders
\cite{Higman52}. \qed

\end{proof}


We conclude that the following type of learner is convergent: for
the empty sample, propose the empty hyperrectangle; for every
non-empty sample $S$, propose a maximal hyperrectangle consistent with
$S$.


\section{Synthesis Problems Modeled as ALFs}\label{sec:examples}
In this section, we list a host of existing synthesis problems and algorithms that can be seen as ALFs. Specifically, we consider examples from the areas of program verification and program synthesis.
We encourage the reader to look up the referenced algorithms to better understand their mapping into our framework.
Moreover, we have new techniques based on ALFs to compute fixed-points in the setting of abstract interpretation using learning.

\subsection{Program Verification} \label{subsec:program_verification}
While program verification itself does not directly relate to
synthesis, most program verification techniques require some form of
help from the programmer before the analysis can be
automated. Consequently, synthesizing objects that replace manual help
has been an area of active research. We here focus on \emph{learning
  loop invariants}. Given adequate
invariants (in terms of pre/post-conditions, loop invariants, etc.), the rest of the verification process can often be completely
automated~\cite{floyd,hoare} using logical constraint-solvers~\cite{z3,cvc4}.
For the purposes of this article, let us consider while-programs with a single loop. Given a pre- and post-condition, assertions, and contracts for functions called, the problem is to find a loop invariant that proves the post-condition and assertions (assuming the program is correct).

\subsubsection{Invariant synthesis using the ICE learning model}
Given a program with a single loop whose loop invariant we want to synthesize, there are \emph{many} inductive invariants that prove
the assertions in the program correct--- these invariants are characterized by the following three properties:
\begin{enumerate*}
	\item that it include the states when the loop is entered the first time,
	\item that it exclude the states that immediately exit the loop and reach the end of the program and not satisfy the post-condition, and
	\item that it is inductive (i.e., from any state satisfying the invariant, if we execute the loop body once, the resulting state is also in the invariant).
\end{enumerate*}
The teacher knows these properties, and must reply to conjectured hypotheses of the learner using these properties. Violation of properties (a) and (b) are usually easy to check using a constraint solver, and will result in a \emph{positive} and \emph{negative} concrete configuration as a sample, respectively. However, when inductiveness fails, the obvious counterexample is a \emph{pair} of configurations, $(x,y)$, where $x$ is in the hypothesis but $y$ is not, and where the program state $x$ evolves to the state $y$ across one execution of the loop body.

The work by Garg et.~al.~\cite{DBLP:conf/cav/0001LMN14} hence proposes what they call the \emph{ICE model} (for implication counterexamples), where the learner learns from positive, negative, and implication counterexamples. The author's claim is that without implication counterexamples, the teacher is stuck when presented a hypothesis that satisfies the properties of being an invariant save the inductiveness property.


From the described components we build an ALF $\A_{\mathrm{ICE}} =
(\cspace, \hypspace, \concrete, \sspace, \consistent)$, where 
$\cspace$ is the set of all subsets of program configurations, the hypothesis space $\hypspace$ is the language used to describe the invariant, and the sample space is defined as follows:
\begin{itemize}
\item A sample is of the form $S = (P, N, I)$, where $P,N$ are sets of
  program configurations (interpreted as positive and negative
  examples), and $I$ is a set of pairs of program configurations
  (interpreted as implications).
\item A set $C \in \cspace$ of program configurations is consistent
  with $(P, N, I)$ if $P \subseteq C$, $N \cap C = \emptyset$, and if
  $(c,c') \in I$ and $c \in C$, then also $c' \in C$.
\item The order on samples is defined by component-wise set inclusion (i.e., $(P,N,I) \sord (P',N',I')$ if $P \subseteq P'$, $N \subseteq N'$, and $I \subseteq I'$).
\item The join is the component-wise union, and $\sbot = (\emptyset,\emptyset,\emptyset)$. 
\end{itemize}
Since this sample space contains implications in addition to the standard positive and negative examples, we refer to it as an \emph{ICE sample space}.

%
%
We can now show that there is a teacher for these ALF instances, because
a teacher can refute any hypothesis made by a learner with a positive, negative,
or implication counterexample, depending on which property of invariants is violated.

\begin{proposition} \label{pro:ice-teacher}
There is a teacher for ALF instances of the form $(\A_{\mathrm{ICE}},
\target_{\mathit{Inv}})$.
\end{proposition}


Furthermore, we can show that having only positive and negative samples precludes the existence of teachers. In fact, we can show that if $\cspace = 2^D$ (for a domain $D$) and the sample space $\sspace$ consists of only positive and negative examples in $D$, then a target set $\target$ has a teacher only if it is defined in terms of excluding a set $B$ and including a set $G$.

\begin{lemma}
Let $C=H=2^D$, $\gamma = id$, $S = \{(P, N) \mid P, N \subseteq D \}$, and
$\kappa((P,N)) = \{ R \subseteq D \mid P \subseteq R \wedge R \cap N = \emptyset \}$.

Let $\target \subseteq C$ be a target. If there exists a teacher for $\target$,
then there must exists sets $B,G \subseteq D$ such that
$\target = \{ R \subseteq D \mid B \cap R = \emptyset \text{ and } G \subseteq R\}$.
\end{lemma}

\begin{proof}
 Assume that there exists a teacher for the target set $\target$, and let $G$ and $B$ be the union of all positive examples and the union of all negative examples, respectively, that the teacher returns. Now, we claim that $\target = \{ R \subseteq D \mid B \cap R = \emptyset \text{ and } R \subseteq G\}$. Towards a contradiction, assume that this is not the case. Then, there exists an $R \in \target$ such that $B \cap R \not = \emptyset$ or $G \not \subseteq R$. If $B \cap R \not = \emptyset$, then there is some $b \in R$ that was returned as a negative counterexample for some hypothesis. Since $R \in \target$ is not consistent with this negative example $b$, this contradicts the requirement that the teacher is honest. Similarly, if $G \not \subseteq R$, then there is some $g \in G$ that was returned as a positive counterexample, which contradicts the teacher's honesty.
\end{proof}

The above proves that positive and negative samples are not sufficient for learning invariants,
as invariants cannot be defined as all sets that exclude a set of states and include a set of states.\par\vskip\baselineskip

There are several ICE-based invariant synthesis formalisms that we can capture. First, Garg~et.~al.~\cite{DBLP:conf/cav/0001LMN14} have considered arithmetic invariants over a set of integer variables  $x_1, \ldots, x_\ell$ of the form $\bigvee_{i=1}^n \bigwedge_{j=1}^{m_i} \bigl( \sum_{k=1}^\ell a_k^{i,j} x_k \leq c^{i,j} \bigr)$, $a_k^{i, j} \in \{{-1}, 0, 1\}$, where the learner is implemented using a constraint solver that finds smallest invariants that fit the sample. This is accurately modeled in our framework, as in the ICE formulation above, with the hypotheses space being the set of all formulas of this form. The fact that Garg~et.~al.'s learner produces smallest invariants makes it an Occam learner in the sense of Section~\ref{sec:occam_learner} and, hence, it converges in finite time.
The approach proposed by Sharma and Aiken~\cite{DBLP:conf/cav/0001A14}, {\sc C2I}, is also an ICE-learner, except that the learner uses \emph{stochastic search} based on a Metropolis Hastings MCMC (Markov chain Monte Carlo) algorithm, which again can be seen as an ALF.

We can also see the work by Garg~et.~al.~\cite{CAVQDA} on synthesizing quantified invariants for linear data structures such as lists and arrays as ALFs. This framework can infer quantified invariants of the form 
\[ \forall y_1, y_2 \colon (y_1 \leq y_2 \leq i) \Rightarrow a[y_1] \leq a[y_2]. \]
However, Garg~et.~al.\ do not represent sets of configurations by means of logical formulas (as shown above) but use an automata-theoretic approach, where a special class of automata, called \emph{quantified data automata (QDAs)}, represent such logical invariants; hence, these QDAs form the hypothesis space in the ALF. The sample space there is also unusual: a sample (modeling a program configuration consisting of arrays or lists) is a \emph{set of valuation words}, where each such word encodes the information about the array (or list) for specially quantified pointer variables pointing into the heap, and where data-formulas state conditions of the keys stored at these locations. 

ALFs can also capture the ICE-framework described by Neider~\cite{daniel_phd}, where invariants are learned in the context of \emph{regular model-checking}~\cite{DBLP:conf/cav/BouajjaniJNT00}. In regular model checking, program configurations are captured using (finite---but unbounded---or even infinite) words, and sets of configurations are captured using finite automata. Consequently, the hypothesis space is the set of all DFAs (over an a~priori chosen, fixed alphabet), and the sample space is an ICE-sample consisting of configurations modeled as words. The learner proposed by Neider constructs consistent DFAs of minimal size and, hence, is an Occam learner that converges in finite time (cf.\ Section~\ref{sec:occam_learner}).
\par\vskip\baselineskip

We now turn to two other invariant-generation frameworks that skirt the ICE model.

\subsubsection{Houdini}
The Houdini algorithm~\cite{houdini} is a learning algorithm for synthesizing invariants that avoids learning from ICE samples.
Given a finite set of \emph{predicates} $P$, Houdini learns an invariant that is expressible as a conjunction of some subset
of predicates (note that the hypothesis space is finite but exponential in the number of predicates). 
Houdini learns an invariant in time \emph{polynomial} in the number of predicates (and in linear number rounds) and
is implemented in the Boogie program verifier~\cite{boogie}. It is widely used 
(for example, used in verifying device drivers~\cite{DBLP:conf/sigsoft/LalQ14,DBLP:conf/cav/LalQL12} and in race-detection in GPU kernels~\cite{DBLP:conf/oopsla/BettsCDQT12}.

The setup here can be modeled as an ALF: we take the concept space $\cspace$ to be all subsets of program configurations, and the hypothesis space $\hypspace$ to be the set of all conjunctions of subsets of predicates in $P$, with the map $\gamma$ mapping each conjunctive formula in $\hypspace$ to the set of all configurations that satisfy the predicates mentioned in the conjunction. We take the sample space to be the ICE sample space, where each sample is a valuation $v$ over $P$ (indicating which predicates are satisfied) and where implication counterexamples are pairs of valuations. 

The Houdini learning algorithm itself is the classical conjunctive learning algorithm for positive and negative samples (see \cite{KearnsV94}), but its mechanics are such that it works for ICE samples as well. More precisely, the Houdini algorithm always creates the \emph{semantically smallest} formula that satisfies the sample (it hence starts with a conjunction of all predicates, and in each round ``knocks off'' predicates that are violated by positive samples returned). Since Houdini always returns the semantically-smallest conjunction of predicates, it will never receive a negative counterexample
(assuming the program is correct and has a conjunctive invariant over $P$). Furthermore, for an implication counterexample $(v,v')$, the algorithm knows that since it proposed the semantically smallest conjunction of predicates, $v$ cannot be made negative; hence it treats $v'$ as a positive counterexample. Houdini converges since the hypothesis space is finite (matching the first recipe for convergence we outlined in Section~\ref{sec:finite_hypothesis_space}); in fact, it converges in linear number of rounds since in each round at least one predicate is removed from the hypothesized invariant.

\subsubsection{Learning invariants in regular model-checking using witnesses}
The learning-to-verify project reported in~\cite{DBLP:conf/icfem/VardhanSVA04,DBLP:conf/fsttcs/VardhanSVA04,DBLP:conf/tacas/VardhanSVA05,DBLP:conf/kbse/VardhanV05} leverages machine learning to the verification of infinite state systems that result from processes with FIFO queues, 
but skirts the ICE model using a different idea. 
The key idea is to view the the identification of the reachable states of such a system as a machine learning problem instead of computing this set iteratively (which, in general, requires techniques such as acceleration or widening). 
In particular, we consider the work of Vardhan~et.~al.~\cite{DBLP:conf/icfem/VardhanSVA04} and show that this is an instantiation of our abstract learning framework. The key idea in this work is to represent configurations as traces through the system and to add a notion of \emph{witness} to this description, resulting in so-called annotated traces. The teacher, when receiving a set of annotated traces, can actually check whether the configurations are reachable based on the witnesses (a witness can be, say, the length of the execution or the execution itself) and, consequently, the model allows learning from such traces directly. 
Indeed, this approach can be modeled by an ALF: the concept space consists of subsets of configurations, the target space consists of the set of reachable configurations, the hypothesis space consists of automata over annotated traces, and the sample space consists of positive and negatively labeled annotated traces.

\subsection{Synthesis of Fixpoints in Abstract Domains}\label{subsec:abstract}

In Section~\ref{subsec:program_verification} we have explained how several
instances of ICE learners fit into our framework. We now provide a
generic technique to model the problem of fixpoint computation in the
setting of abstract interpretation using learning.

\newcommand{\D}{\mathcal{D}}
\newcommand{\Dhat}{\widehat{\mathcal{D}}}

We assume the setting of abstract interpretation~\cite{cc77} with 
\begin{itemize}
\item a concrete domain $(\D, \subseteq, \bot, \cup)$ and an abstract
  domain $(\Dhat, \sqsubseteq, \bot, \sqcup)$, which both have a join
  lattice structure, and
\item a Galois connection between these two, given by two monotone
  functions $\gamma: \Dhat \rightarrow \D$ (the concretization
  function), and $\alpha: \D \rightarrow \Dhat$ (the abstraction
  function), with $\alpha(X) \sqsubseteq \widehat{X} \Leftrightarrow X
  \subseteq \gamma(\widehat{X})$ for all $X \in \D$ and $\widehat{X} \in
  \Dhat$.

\end{itemize}
The concrete domain usually describes the semantics of a program and
comes with an increasing transformer $F: \D \rightarrow \D$ that
captures the behaviour of the program (increasing means that
$X \subseteq X$ for each $X$). The abstract domain is used to model
the aspects of the program that one is interested in.

A specification is given in terms a set $B \subseteq \D$ of bad
concrete elements. The goal is to find a fixpoint of $F$ that is not
above any bad element, i.e., an element $X \in \D$ with $F(X) = X$,
and $Y \not\subseteq X$ for each $Y \in B$. We refer to such fixpoints
as \emph{adequate fixpoints} because they show that the program cannot
reach a bad state.

\paragraph{Synthesis of Precise Abstract Transformers}
The general idea of abstract interpretation is to do the fixpoint
computation on the abstract domain instead of the concrete
domain. This is done using an abstract transformer
$\widehat{F}: \Dhat \rightarrow \Dhat$ that overapproximates the concrete
transformer, in the sense that
$\alpha(F(\gamma(\hat{X}))) \sqsubseteq \hat{F}(\hat{X})$ for each
$\hat{X} \in \Dhat$. The best abstract transformer is the one where
equality holds instead of inclusion.

Since manually designing a good abstract transformer can be
difficult, Thakur et al.~\cite{ThakurLLR15} propose an automatic method to
synthesize the abstract post of a given abstract element for the
concrete domain consisting of sets of program configurations $\D
= 2^D$. 

We can model the algorithm proposed by Thakur et al.~\cite{ThakurLLR15} as a 
synthesis using learning in an ALF that uses the concrete domain $\D$
as concept space, the abstract domain $\Dhat$ as hypothesis space
together with the existing concretization function $\gamma$ as the
concretization function for the ALF. Since we
are interested in elements above $\alpha(F(\gamma(\hat{X})))$, the
sample space simply consists of positive examples, that is, a sample
is a finite set of elements of $D$. A hypothesis is consistent with a
sample if it contains all the elements from the sample. 

For an abstract element $\hat{X}$, the target specification is given
by $\target = \{T \subseteq
D \mid \alpha(F(\gamma(\hat{X}))) \sqsubseteq T\}$. 

If the learner proposes some hypothesis $H$ (which is an abstract
element), then the teacher can check whether $\gamma(H)$ is a superset
of $F(\gamma(\widehat{X}))$ and if not, return a positive example from
$F(\gamma(\widehat{X})) \setminus \gamma(H)$. This is precisely what happens
(at an abstract level) in the procedure proposed in \cite{ThakurLLR15}.

One should note here that in case the abstract domain has a maximal
element, a learner could always propose this element because it can be
sure that it is a target element if the target is realizable. The idea
is to come up with a learner that computes a better solution than just
the maximum. In \cite{ThakurLLR15} this is done by starting with the
empty hypothesis, and then for each positive example $S$ returned by
the teacher, applying the abstraction function $\alpha$ on it and taking the join
of the resulting abstract element with the current hypothesis. In this
way, the learner ensures that the hypothesis is always below the set
$\alpha(F(\gamma(\widehat{X})))$. So if the learner converges, then it
computes the best abstract transformer on $\widehat{X}$.

\paragraph{Synthesis of Adequate Fixpoints.}
We now describe a generic way for designing an ALF that models the problem
of finding an adequate fixpoint as a synthesis problem that uses learning
from samples. This can be seen as a generalization of finding loop invariants
in an abstract domain (but without computing the abstract transformers).


The ALF we propose has the concrete domain $\D$ as concept space, the
abstract domain $\Dhat$ as hypothesis space together with the existing
concretization function $\gamma$. The sample space of an ALF is, in
general, designed for a specific class of target specifications (in
order to guarantee the existence of a teacher). In the following, we
describe how to construct such a sample space for adequate fixpoints
as targets. 

The invariants defined in Section~\ref{subsec:program_verification} can be
seen as adequate fixpoints for the  transformer $F$ and the set $B$
being defined pointwise on single configurations (with concrete domain
consisting of sets of program configurations). This pointwise
definition is the reason why the sample space can be build up by using
single configurations as positive and negative examples, and pairs of
configurations as implications.

We replace the pointwise definition by a definition of in terms of a
class $\mathcal{R} \subseteq \D$ of representative concrete
elements. So the definition below is satisfied by the class of
singleton sets in case of the concrete domain consisting of sets of
program configurations. It intuitively states that $\mathcal{R}$ is
rich enough to prove that a hypothesis is not a fixpoint.

%

\begin{definition}\label{def:representative-set}
A set $\mathcal{R} \subseteq \D$ is a \emph{representative set} (for
the transformer $F$) if for each $\hat{X} \in \Dhat$ and $X
:= \gamma(\hat{X})$ such that $F(X) \not= X$, there are
$Y,Y' \in \mathcal{R}$ with $Y \subseteq X$, $Y'
\subseteq F(Y)$, and $Y' \not\subseteq X$.\qed
\end{definition}

It is not hard to see that $\mathcal{D}$ is always representative set.
And in fact that the set of all elements of the form $\gamma(H)$ and $F(\gamma(H))$,
where $H \in \hypspace$ also forms a representative set.

We consider target specifications of adequate fixpoints that can be
expressed in terms of $\mathcal{R}$. We say that
$\target \subseteq \D$ is an $\mathcal{R}$-specification of adequate
fixpoints if there is a set $B \subseteq \mathcal{R}$ such that
$\target =\{X \in \D \mid F(X) = X \mbox{ and } Y \not\subseteq
X \mbox{ for all } Y \in B \}$.


For this class of target specifications, we let the sample space
$\sspace_\mathcal{R}$ consist of ICE samples over $\mathcal{R}$, that
is, of triples $(P,N,I)$ of finite set $P,N \subseteq \mathcal{R}$,
and a finite set $I \subseteq \mathcal{R} \times \mathcal{R}$. A
concept $X \in D$ is consistent with a sample $(P,N,I)$ if
$Y \subseteq X$ for all $Y \in P$, $Y \not\subseteq X$ for all $Y \in
N$, and if $Y \subseteq X$, then $Y' \subseteq X$ for all $(Y,Y') \in
I$.
This consistency relation is denoted by $\consistent_\mathcal{R}$.  In
analogy to Proposition~\ref{pro:ice-teacher} we can show that this
sample space is expressive enough to guarantee the existence of a
teacher.
\begin{proposition}  \label{pro:abstract-ice-teacher}
There is a teacher for ALF instances of the form $(\A,\target)$ with
$\A = (\D, \Dhat, \gamma, \sspace_{\mathcal{R}},
\consistent_{\mathcal{R}})$ and an $\mathcal{R}$-specification of
adequate fixpoints $\target$.
\end{proposition}
%
%
%
%

This provides a generic way of setting up a learning scenario for
abstract interpretation, and thus provides a powerful tool for
understanding the requirements for the application and development of
machine learning algorithms for the synthesis of adequate fixpoints.

\subsection{Program Synthesis}
In this section, we study several examples of learning-based program synthesis, which include synthesizing program expressions,
expressions to be plugged in program sketches, snippets of programs, etc., and show how they can be modeled as ALFs.

\subsubsection{End-user synthesis from examples: Flashfill}
One application of synthesis is to use it to help end-users to program using examples. A prime example of this is \textsc{Flashfill} by Gulwani et al~\cite{DBLP:conf/popl/Gulwani11}, where the authors show how string manipulation macros from user-given input-output examples can be synthesized in the context of Microsoft Excel spreadsheets. Flashfill can be seen as an ALF: the concept space consists of all functions from strings to strings, the hypothesis space consists of all string manipulation macros, and the sample space consists of a sets of input-output examples for such functions. The consistency relation $\kappa$ maps each sample to all functions that agree with the sample. The role of the teacher is played by the \emph{user}: the user has some function in mind and gives new input-output examples whenever the learner returns a hypothesis that she is not satisfied with. The learning algorithm here is based on version-space algebras (which, intuitively, compactly represents \emph{all} possible macros with limited size that are consistent with the sample) and in each round proposes a simple macro from this collection.

\subsubsection{Completing sketches and the SyGuS solvers}
The sketch-based synthesis approach~\cite{armando_thesis} is another prominent synthesis application, where programmers write partial programs with holes and a system automatically synthesizes expressions or programs for these holes so that a specification (expressed using input-output pairs or logical assertions) is satisfied. The key idea here is that given a sketch with a specification, we need expressions for the holes such that \emph{for every possible input}, the specification holds. This roughly has the form
$\exists \vec{e}. \forall \vec{x} \psi (\vec{e}, \vec{x})$, where $\vec{e}$ are the expressions to synthesize and
$\vec{x}$ are the inputs to the program.

The Sketch system works by
\begin{enumerate*}
	\item unfolding loops a finite number of times, hence, bounding the length of executions, and
	\item encoding the choice of expressions $\vec{e}$ to be synthesized using bits (typically using templates and representing integers by a small number of bits). 
\end{enumerate*}
For the synthesis step, the Sketch system implements a CEGIS (counterexample guided synthesis) technique using SAT solving, whose underlying idea is to learn the expressions from examples using only a SAT solver. The CEGIS technique works in rounds: the learner proposes hypothesis expressions and the teacher checks whether 
$\forall \vec{x} \psi (\vec{e}, \vec{x})$ holds (using SAT queries) and if not, returns a valuation for $\vec{x}$ as a counterexample. Subsequently, the learner asks, again using a SAT query, whether there exists a valuation for the bits encoding the expressions such that $\psi (\vec{e}, \vec{x})$ holds for every valuation of $\vec{x}$ returned by the teacher thus far; the resulting expressions are the hypotheses for the next round. Note that the use of samples avoids quantifier alternation both in the teacher and the learner.

The above system can be modeled as an ALF. The concept space consists of tuples of functions modeling the various expressions to synthesize, the hypothesis space is the set of expressions (or their bit encodings), the map $\concrete$ gives meaning to these expressions (or encodings), and the sample space can be seen as the set of \emph{grounded formulae} of the form $\psi(\vec{e}, \vec{v})$ where the variables $\vec{x}$ have been substituted with a concrete valuation. The relation $\consistent$ maps such a sample to the set of all expressions $\vec{f}$ such that the formulas in the sample all evaluate to true if $\vec{f}$ is substituted for $\vec{e}$. The Sketch learner can be seen as a learner in this ALF framework that uses calls to a SAT solver to find hypothesis expressions consistent with the sample.
Since expressions are encoded by a finite number of bits, the hypothesis space is finite, and the Sketch learner converges in finite time (cf.\ Section~\ref{sec:finite_hypothesis_space}).

The SyGuS format~\cite{DBLP:conf/fmcad/AlurBJMRSSSTU13} is a competition format for synthesis, and extends the Sketch-based formalism above to SMT theories, with an emphasis on syntactic restrictions for expressions. More precisely, SyGuS specifications are parameterized over a background theory $T$, and an instance is a pair $(G, \psi(\vec{f}))$ where $G$ is a grammar that imposes syntactic restrictions for  functions (or expressions) written using symbols of the background theory, and $\psi$ is a formula, again in the theory $T$, including function symbols $\vec{f}$; the functions $\vec{f}$ are
typed according to domains of $T$. The goal is to find functions $\vec{g}$ for the symbols $\vec{f}$ in the syntax $G$ such that $\psi$ holds. The competition version further restricts $\psi$ to be of the form $\forall \vec{x} \psi'(\vec{f}, \vec{x})$ where $\psi'$ is a quantifier-free formula in a decidable SMT theory---this way, given a hypothesis for the functions $\vec{f}$, the problem of checking whether the functions meet the specification is decidable.

There have been several solvers developed for SyGuS (cf.\ the first SyGuS competition~\cite{DBLP:conf/fmcad/AlurBJMRSSSTU13,DBLP:series/natosec/AlurBDF0JKMMRSSSSTU15}), and all of them are in fact learning-based (i.e., CEGIS) techniques. In particular, three solvers have been proposed: an enumerative solver, a constraint-based solver, and a stochastic solver. All these solvers can be seen as ALF instances: the concept space consists of all possible tuples of functions over the appropriate domains and the hypothesis space is the set of all functions allowed by the \emph{syntax} of the problem (with the natural $\gamma$ relation giving its semantics). All three solvers work by generating a tuple of functions such that $\forall \vec{x} \psi'(\vec{f}, \vec{x})$ holds for all valuations of $\vec{x}$ given by the teacher thus far.
The enumerative solver enumerates functions until it reaches such a function, the stochastic solver searches the space of functions randomly using a measure that depends on how many samples are satisfied till it finds one that satisfies the samples, and the constraint-based solver queries a constraint-solver for instantiations of template functions so that the specification is satisfied on the sample valuations.
Both the enumerative and the constraint-solver are Occam learners and, hence, converge in finite time.

Note that the learners \emph{know} $\psi$ in this scenario. However, we can model SyGuS as ALFs by taking the sample space to be grounded formulas $\psi'(\vec{f}, \vec{v})$ consisting of the specification with particular values $\vec{v}$ substituted for $\vec{x}$. The learners can now be seen as learning from these samples, without knowledge of $\psi$ (similar to the modeling of Sketch above).

We would like to emphasize that this embedding of SyGuS as an ALF clearly showcases the difference between different synthesis approaches (as mentioned in the introduction). For example, invariant generation can be done using learning either by means of ICE samples (see Section~\ref{subsec:program_verification}) or modeled as a SyGuS problem. However, it turns out that the sample spaces (and, hence, the learners) in the two approaches are \emph{very different}! In ICE-based learning, samples are only single configurations (labeled positive or negative) or pairs of configurations, while in a SyGuS encoding, the samples are grounded formulas that encode the entire program body, including instantiations of universally quantified variables intermediate states in the execution of the loop.

\subsubsection{Machine-learning based approaches to synthesis}
One can implement the \emph{passive} machine-learning algorithm to synthesize hypotheses from samples, in order to build a synthesis engine (along with an appropriate teacher that can furnish such samples). Recent work by Garg et.~al.~\cite{ICEML} proposes an algorithm for synthesizing invariants in the ICE-framework using machine learning classifiers (decision trees) that can be viewed as an ALF.

\subsubsection{Synthesizing guarded affine functions}
Recent work~\cite{alchemist} explores the synthesis of \emph{guarded affine functions} from a sample space that consists of information of the form $f(\vec{s})=t$, where $\vec{s}$ and $\vec{t}$ are integers. The learner here uses a combination of computational geometry techniques and decision tree learning, and can also be modeled as an ALF. Notice that this sample space  precisely matches the sample space for deobfuscation problems
(where the teacher can return counterexamples of this form -- see Example at the end of Section~\ref{sec:alf} on page~\pageref{ex:deobfuscation} -- using the program being deobfuscated). Consequently, the
learner in Alchemist~\cite{alchemist} can be used for deobfuscating programs that compute guarded affine functions from tuples of integer inputs
to integers (like the ``multiply by 45'' example  in~\cite{DBLP:conf/icse/JhaGST10}).

\subsubsection{Other synthesis engines}
There are several algorithms that are self-described as CEGIS frameworks, and, hence, can be modeled using ALFs. For example, synthesizing loop-free programs~\cite{DBLP:conf/pldi/GulwaniJTV11}, synthesizing synchronizing code for concurrent programs~\cite{scheduling} (in this work, the sample space consists of abstract concurrent partially-ordered traces), work on using synthesis to \emph{mine specifications}~\cite{DBLP:conf/hybrid/JinDDS13}, synthesizing bit-manipulating programs and deobfuscating programs~\cite{DBLP:conf/icse/JhaGST10} (here, the use of separate I/O-oracle can be modeled as the teacher returning the output of the program together with a counterexample input), superoptimization~\cite{DBLP:conf/asplos/Schkufza0A13}, deductive program repair~\cite{DBLP:conf/cav/KneussKK15}, synthesis of recursive functional programs over unbounded domains~\cite{DBLP:conf/oopsla/KneussKKS13}, as well as synthesis of protocols using enumerative CEGIS techniques~\cite{DBLP:conf/pldi/UdupaRDMMA13}.

\section{Variations and Limitations of the Framework} \label{app:variations}
In this section we discuss some variations and limitations of our
framework. We start by briefly discussing a variation of our framework
that omits the concept space.

\subsection{Omitting the Concept Space}
We believe that, for a clean modeling of a synthesis problem, one
should specify the concept space $\cspace$. This makes it possible to
compare different synthesis approaches that work with different
representations of hypotheses and maybe different types of samples
over the same underlying concept space.

However, for the actual learning process, the concept space
itself is not of great importance because the learner proposes elements from
the hypothesis space, and the teacher returns an element from the
sample space. The concept space only serves as a semantic space that gives
meaning to hypotheses (via the concretization function $\concrete$),
and to the samples (via the consistency relation $\consistent$). 

Therefore, it is possible to omit the concept space from an ALF, and
to directly specify the consistency of samples with hypotheses. Such a
reduced ALF would then be of the form $\A = (\hypspace, \sspace,
\consistent)$ with a function $\consistent: \sspace \rightarrow
2^\hypspace$. In the original framework, this corresponds to the
function $\consistent_\hypspace$ defined by
$\consistent_\hypspace(S) = \concrete^{-1}\consistent(S)$.

To create ALF instances, the target specification is also directly
given as a subset of the hypothesis space $\target
\subseteq \hypspace$. All the other definitions can be adapted directly
to this framework.

\subsection{Limitations}
The ALF framework we develop in this paper is not meant to capture every 
existing method that uses learning from samples. There are several synthesis techniques 
that use grey-box techniques (a combination of black-box learning from samples and by utilizing
the specification of the target directly in some way) or use query models (where they query the
teacher for various aspects of the target set). 

For instance, there are active iterative learning scenarios
in which the learner can ask other types of questions to the teacher than just proposing
hypotheses that are then accepted or refuted by the teacher. One
prominent scenario of this kind is Angluin's active learning of DFAs
\cite{Angluin87}, where the learner can ask \emph{membership queries}
and \emph{equivalence queries}. (The equivalence queries correspond to
proposing a hypothesis, as in our framework, which is then refuted
with a counterexample if it is not correct.)
Such learning scenarios for synthesis are used, for example, in
\cite{AlurCMN05} for the synthesis of interface specifications for Java classes, 
and in \cite{PasareanuGBCB08} for
automatically synthesizing assumptions for assume-guarantee reasoning.
Our framework does not have a mechanism for directly modeling such
queries. 
The ALF framework that we have presented
is intentionally a simpler framework by design that captures and cleanly models emerging synthesis 
procedures in the literature where the learner only proposes hypotheses and learns
from samples the teacher provides in terms of samples to show that the 
hypothesis is wrong. The learner in our framework, being a completely passive learner 
(as opposed to an active learner), can also be implemented by the variety of scalable passive
machine-learning algorithms in vogue~\cite{mitchell}.  A clean extension of ALFs to query
settings and grey-box settings would be an interesting future direction to pursue.


%

\section{Conclusions}
\vspace{-.5\baselineskip}
We have presented an abstract learning framework for synthesis that encompasses several
existing techniques that use learning or counter-example guided inductive synthesis to 
create objects that satisfy a specification.
We were motivated by abstract interpretation~\cite{cc77} and how it
gives a general framework and notation for verification; our formalism is an attempt at such a generalization
for learning-based synthesis. The conditions we have proposed that the abstract concept spaces, hypotheses spaces,
and sample spaces need to satisfy to define a learning-based synthesis domain seem to be cogent
and general in forming a vocabulary for such approaches. We have also addressed various strategies
for convergent synthesis that generalizes and extends existing techniques (again, in a similar
vein as to how widening and narrowing in abstract interpretation give recipes for building convergent
algorithms to compute fixed-points). We believe that the notation and general theorems herein
will bring more clarity, understanding, and reuse of learners in synthesis algorithms.
%

\vspace{-.25\baselineskip}
\paragraph{\normalfont\bfseries Acknowledgements:} This work was partially supported by NSF Expeditions
in Computing ExCAPE Award \#1138994.
\vspace{-.5\baselineskip}

\enlargethispage{\baselineskip}
\bibliographystyle{splncs03}
\bibliography{alf}

\end{document}